\documentclass[12pt,twoside]{article}

 \usepackage{float}
\usepackage{graphicx}
\usepackage{epstopdf}
\usepackage{graphicx}
\usepackage{epic}
\usepackage{multirow}
\usepackage{tikz}
\usepackage{xcolor}
\usetikzlibrary{arrows,shapes,chains}

\renewcommand{\paragraph}{\roman{paragraph}}
\usepackage[a4paper]{geometry}
\makeatletter
\renewcommand\title[1]{\gdef\@title{\reset@font\Large\bfseries #1}}
\renewcommand\section{\@startsection {section}{1}{\z@}%
                                   {-3.5ex \@plus -1ex \@minus -.2ex}%
                                   {2.3ex \@plus.2ex}%
                                   {\normalfont\large\bfseries}}
\renewcommand\subsection{\@startsection{subsection}{2}{\z@}%
                                     {-3ex\@plus -1ex \@minus -.2ex}%
                                     {1.5ex \@plus .2ex}%
                                     {\normalfont\normalsize\bfseries}}
\renewcommand\subsubsection{\@startsection{subsubsection}{3}{\z@}%
                                     {-2.5ex\@plus -1ex \@minus -.2ex}%
                                     {1.5ex \@plus .2ex}%
                                     {\normalfont\normalsize\bfseries}}

\def\@runningauthor{}\newcommand{\runningauthor}[1]{\def\runningauthor{#1}}
\def\@runningtitle{}\newcommand{\runningtitle}[1]{\def\runningtitle{#1}}

\renewcommand{\ps@plain}{%
\renewcommand{\@evenhead}{\footnotesize\scshape \hfill\runningauthor\hfill}
\renewcommand{\@oddhead}{\footnotesize\scshape \hfill\runningtitle\hfill}}

\newcommand{\F}{\mathbb{F}}

\g@addto@macro\bfseries{\boldmath}

\makeatother

\usepackage{amsthm,amsmath,amssymb}
\usepackage{cite}
\usepackage{graphicx}

\usepackage[colorlinks=true,citecolor=black,linkcolor=black,urlcolor=blue]{hyperref}

\theoremstyle{plain}
\newtheorem{theorem}{Theorem}[section]

\newtheorem{lem}[theorem]{Lemma}

\newtheorem{prop}[theorem]{Proposition}

\theoremstyle{definition}
\newtheorem{definition}[theorem]{Definition}
\newtheorem{example}[theorem]{Example}
\newtheorem{conjecture}[theorem]{Conjecture}

\theoremstyle{remark}
\newtheorem{remark}[theorem]{Remark}

\runningauthor{}

\date{}

\begin{document}

\title{Two conjectures on the largest minimum distances of binary self-orthogonal codes with dimension 5}
\author{Minjia Shi\thanks{smjwcl.good@163.com}, Shitao Li\thanks{lishitao0216@163.com}, Jon-Lark Kim\thanks{jlkim@sogang.ac.kr}
\thanks{Minjia Shi and Shitao Li are with School of Mathematical Sciences, Anhui University, Hefei, China.
Jon-Lark Kim is with Sogang University, Seoul, South Korea.}}

\date{}
    \maketitle

\begin{abstract}
The purpose of this paper is to solve the two conjectures on the largest minimum distance $d_{so}(n,5)$ of a binary self-orthogonal $[n,5]$ code proposed by Kim and Choi (IEEE Trans. Inf. Theory, 2022). The determination of $d_{so}(n,k)$ has been a fundamental and difficult problem in coding theory because there are too many binary self-orthogonal codes as the dimension $k$ increases. Recently, Kim et al. (2021) considered the shortest self-orthogonal embedding of a binary linear code, and many binary optimal self-orthogonal $[n,k]$ codes were constructed for $k=4,5$. Kim and Choi (2022) improved some results of Kim et al. (2021) and made two conjectures on $d_{so}(n,5)$. In this paper, we develop a general method to determine the exact value of $d_{so}(n,k)$ for $k=5,6$ and show that the two conjectures made by Kim and Choi (2022) are true.
\end{abstract}

\section{Introduction}
In coding theory, self-orthogonal (for short, SO) codes over finite fields form an important class of codes which are asymptotically good \cite{SO-good} and have been extensively studied.
Some typical linear codes are SO, for example, the binary simplex
$[7, 3, 4]$ code $S_3$, the extended binary $[8, 4, 4]$ Hamming code, the extended binary $[24, 12, 8]$ Golay code, and the extended ternary $[12, 6, 6]$ Golay code.
It is well-known that they have close connections with group theory \cite{C-1}, design theory \cite{t-design}, and lattice theory \cite{B-1,C-1,H-1}. They also have been employed to construct quantum codes \cite{Quantum-1,Quantum-2}.
Therefore, the construction and classification of SO codes has been a hot topic.

Let $d(n,k)$ denote the largest minimum distance among all binary $[n, k]$ codes and
$d_{so}(n,k)$ denote the largest minimum distance among all binary $[n, k]$ SO codes. A binary $[n,k]$ SO code is optimal SO if it has the minimum distance $d_{so}(n,k)$. One of fundamental topics in coding theory is to determine the largest minimum distance of SO codes for various lengths and dimensions. In 2006, Bouyukliev et al. \cite{SO-40} completed the characterization of binary optimal SO codes for parameters up to length 40 and dimension 10, and gave the exact value of $d_{so}(n,3)$. In \cite{Li-Xu-Zhao}, Li, Xu, and Zhao partially characterized four-dimensional optimal SO codes by systems of linear equations.
Later, Kim et al. \cite{Kim-embedding} completely determined the remaining cases and partially characterized the exact value of $d_{so}(n,5)$ by embedding linear codes into SO codes.
On the classification of binary SO codes, the reader is referred to
\cite{Kim-15,Li-Xu-Zhao,Pless-1,Shi-1} for recent papers.

Recently, Kim and Choi \cite{Kim-SO} constructed many new optimal binary SO codes of dimensions 5 and 6 and gave two conjectures.
In this paper, we consider the next step of \cite{SO-40,Kim-embedding,Kim-SO,Li-Xu-Zhao}.
 Using binary SO codes related to the simplex codes, we determine the exact values of $d_{so}(n,5)$ and $d_{so}(n,6)$, which solves the two conjectures proposed by Kim and Choi \cite{Kim-SO}, and furthermore completely generalize some results of the paper.

The paper is organized as follows. In Section 2, we give some notations and preliminaries.
In Section 3, we study some properties of binary SO codes related to the simplex codes. In Section 4, we characterize the exact value of $d_{so}(n,5)$. In Section 5, we characterize the exact value of $d_{so}(n,6)$. In Section 6, we conclude the paper.

\section{Preliminaries}
Let $\F_2$ denote the finite field with 2 elements. For any ${\bf x}\in \F_2^n$, the {\em Hamming weight} of ${\bf x}$ is the number of nonzero components of ${\bf x}$. The {\em Hamming distance} between two vectors ${\bf x}, {\bf y}\in \F_2^n$ is defined to be the number of coordinates in which ${\bf x}$ and ${\bf y}$ differ. The {\em minimum (Hamming) distance} of a code is the smallest Hamming distance between distinct codewords.

An binary linear $[n,k,d]$ code $C$ is a $k$-dimensional subspace of $\F_2^n$. where $d$ is the minimum Hamming distance of $C$.
The dual code $C^{\perp}$ of a binary linear $[n,k]$ code $C$ is defined as
$$C^{\perp}=\{\textbf y\in \F_2^n~|~\langle \textbf x, \textbf y\rangle=0, {\rm for\ all}\ \textbf x\in C \},$$
where $\langle \textbf x, \textbf y\rangle=\sum_{i=1}^n x_iy_i$ for $\textbf x = (x_1,x_2, \ldots, x_n)$ and $\textbf y = (y_1,y_2, \ldots, y_n)\in \F_2^n$.

\begin{definition}
A binary linear code $C$ is {\em self-orthogonal} (SO) if $C\subseteq C^\perp$.
\end{definition}

It is well-known that the Griesmer bound \cite[Chap. 2, Section 7]{Huffman} on a linear $[n,k,d]$
code over $\F_2$ is given by $n\geq \sum_{i=0}^{k-1}\left\lceil \frac{d}{2^i}\right\rceil$, where $\lceil a \rceil$ is the least integer greater than or equal to a real number $a$.
A binary $[n,k]$ SO code is {\em optimal} if it has the largest minimum distance among all binary $[n,k]$ SO codes.

A vector $x = (x_1, x_2, \ldots , x_n)\in \F^n_2$
is {\em even-like} if $\sum_{i=1}^nx_i=0$ and is {\em odd-like} otherwise.
A binary code is said to be {\em even-like} if it
has only even-like codewords, and is said to be {\em odd-like} if it
is not even-like.

\begin{remark}
Since the self-orthogonality, the value $d_{so}(n,k)$ is always even. In addition, the best possible minimum distance of a binary $[n,k]$ SO code is $2\left\lfloor d(n,k)/2\right\rfloor$, that is to say, $d_{so}(n,k)\leq 2\left\lfloor d(n,k)/2\right\rfloor$.
\end{remark}

\section{Binary SO codes related to the simplex codes}

Assume that $S_k$ is a matrix whose columns are all nonzero vectors in $\F_2^k$.
It is well-known that $S_k$ generates a binary simplex code, which is an one-weight SO $[2^k-1,k,2^{k-1}]$ Griesmer code for $k\geq 3$ \cite{Huffman}.
The following lemma shows that we can construct a family of SO codes from a SO code.

\begin{lem}\label{lemma-k-3}
Assume that $S_k$ is a matrix whose columns are all nonzero vectors in $\F_2^k$ for $k\geq 3$. Let $C$ be a binary $[n,k,d]$ linear code with the generator matrix $G$. Then $C$ is SO if and only if $C'$ with the following matrix
$$G'=[\underbrace{S_k|\cdots|S_k}_m|G]$$
is a binary $[m(2^k-1)+n,k,2^{k-1}m+d]$ SO code.
\end{lem}
\begin{proof}
It is well-known that $S_k$ generates a binary simplex code, which is an one-weight SO $[2^k-1,k,2^{k-1}]$ Griesmer code. So
$$G'G'^T=GG^T.$$
Therefore, $C$ is SO code if and only if $C'$ is SO. Since $C$ has the minimum distance $d$, $C'$ has the minimum distance at least $d+2^{k-1}m$. Since the simplex code is a one-weight code, there is a codeword of weight $d+2^{k-1}m$ in $C'$. The converse is also true. This completes the proof.
\end{proof}

\begin{example}
Let $C$ be a binary $[8,4,4]$ SO code with the generator matrix
$$\left(\begin{array}{c}
1 0 0 0 1 1 1 0\\
0 1 0 0 1 1 0 1\\
0 0 1 0 1 0 1 1\\
0 0 0 1 0 1 1 1
\end{array}\right).$$
Then the following matrix
$$\left(\begin{array}{c|c}
      \begin{array}{c}
                    1 0 0 0 1 0 0 1 1 0 1 0 1 1 1\\
0 1 0 0 1 1 0 1 0 1 1 1 1 0 0\\
0 0 1 0 0 1 1 0 1 0 1 1 1 1 0\\
0 0 0 1 0 0 1 1 0 1 0 1 1 1 1\\
                  \end{array}
&
  \begin{array}{c}
     1 0 0 0 1 1 1 0\\
0 1 0 0 1 1 0 1\\
0 0 1 0 1 0 1 1\\
0 0 0 1 0 1 1 1
    \end{array}
      \end{array}\right)$$
  generates a binary $[23,4,12]$ SO code.
\end{example}

\begin{remark}
The lemma \ref{lemma-k-3} shows that we can construct a binary $[n+2^k-1,k,d+2^{k-1}]$ SO code from a binary $[n,k,d]$ SO code. But the converse may not be true. For example, there exists a binary $[45,5,22]$ SO code \cite{Kim-SO}, but there are no binary $[14,5,6]$ SO codes \cite{SO-40}. In fact, the result holds when $2d-n\geq 0$ (see Lemma 3.6 in \cite{AHS-BLCD}). That is to say, if $2d-n\geq 0$ and there exists a binary $[n+2^k-1,k,d+2^{k-1}]$ SO code, then there is a binary $[n,k,d]$ SO code.
\end{remark}

\section{Two conjectures on binary optimal $[n,5]$ self-orthogonal codes}

Kim and Choi propose two conjectures in \cite{Kim-SO}. In this subsection, we prove the conjectures using the above method.

\begin{theorem}{\rm \cite[Conjecture 19]{Kim-SO}}
For $n\geq 32$, if $n\equiv 14,22,29~({\rm mod}~31)$, then $d_{so}(n,5)=d(n,5)$, i.e., there exists an $[n,5,d(n,5)]$ SO code.
\end{theorem}

\begin{proof}
(i) When $n=31m+14$ and $k=5$, the largest minimum distance $d$ satisfying the Griesmer bound is $d=16m+6$.
When $m=1$, there exists a binary $[45,5,22]$ SO code \cite{Kim-SO}. By Lemma \ref{lemma-k-3}, there is a binary $[31m+14,5,16m+6]$ SO code for $m\geq 1$. Hence
\begin{align*}
  16m+6\leq d_{so}(31m+14,5)
    \leq d(31m+14,5)\leq d= 16m+6.
\end{align*}
This implies that $d_{so}(31m+14,5)=d(31m+14,5)=16m+6$, i.e., $d_{so}(n,5)=d(n,5)$ for $n\geq 32$ and $n\equiv 14~({\rm mod}~31)$.

(ii) When $n=31m+22$ and $k=5$, the largest minimum distance $d$ satisfying the Griesmer bound is $d=16m+10$.
When $m=1$, there exists a binary $[53,5,26]$ SO code \cite{Kim-SO}. According to Lemma \ref{lemma-k-3}, there is a binary $[31m+22,5,16m+10]$ SO code for $m\geq 1$. Hence
\begin{align*}
  16m+10\leq d_{so}(31m+22,5)
  \leq  d(31m+22,5)\leq d=16m+10.
\end{align*}
This implies that $d_{so}(31m+22,5)=d(31m+22,5)=16m+10$, i.e., $d_{so}(n,5)=d(n,5)$ for $n\geq 32$ and $n\equiv 22~({\rm mod}~31)$.

(iii) When $n=31m+29$ and $k=5$, the largest minimum distance $d$ satisfying the Griesmer bound is $d=16m+14$.
When $m=1$, there exists a binary $[60,5,30]$ SO code \cite{Kim-SO}. According to Lemma \ref{lemma-k-3}, there is a binary $[31m+29,5,16m+14]$ SO code for $m\geq 1$. Hence
\begin{align*}
  16m+14\leq   d_{so}(31m+29,5)
  \leq  d(31m+29,5)\leq d=16m+14.
\end{align*}
This implies that $d_{so}(31m+29,5)=d(31m+29,5)=16m+14$, i.e., $d_{so}(n,5)=d(n,5)$ for $n\geq 32$ and $n\equiv 29~({\rm mod}~31)$.
\end{proof}

In order to solve Conjecture 20 in \cite{Kim-SO}, we introduce an interesting and useful lemma, which was proved in \cite[Lemma 3.1]{Li-Xu-Zhao}.

\begin{lem}{\rm\cite{Li-Xu-Zhao}}\label{lem-Li-Xu-Zhao}
If there is a binary $[n,k+1,2m]$ SO code, then there is an even-like binary $\left[n-2m,k,2\lceil \frac{m}{2}\rceil\right]$ linear code.
\end{lem}

\begin{theorem}{\rm \cite[Conjecture 20]{Kim-SO}}
If $n=14,21,22,28,29$ or if $n\geq 32$ and $n\equiv 6,13,21,28~({\rm mod}~31)$, then $d_{so}(n,5)=d(n,5)-2$, i.e., there are no $[n,5,d(n,5)]$ SO codes.
\end{theorem}

\begin{proof}
By Table 1 in \cite{SO-40}, there is no binary $[n,5,d(n,5)]$ SO code for $n=14,21,22,28,29$.
From the Database \cite{codetables}, there are binary $[37,5,18]$, $[44,5,22]$, $[52,5,26]$ and $[59,5,30]$ linear codes. Similar to Lemma \ref{lemma-k-3}, there are
binary $[31m+6,5,16m+2]$, $[31m+13,5,16m+6]$, $[31m+21,5,16m+10]$ and $[31m+28,5,16m+14]$ linear codes for some integer $m\geq 1$. Combined with the Griesmer bound, we have
$$\begin{array}{ll}
    d(31m+6,5)=16m+2, &
     d(31m+13,5)=16m+6, \\
    d(31m+21,5)=16m+10,&
     d(31m+28,5)=16m+14.
  \end{array}
$$

If there are binary $[31m+6,5,16m+2]$, $[31m+13,5,16m+6]$, $[31m+21,5,16m+10]$ and $[31m+28,5,16m+14]$ SO codes for some integer $m\geq 1$, then it follows from Lemma \ref{lem-Li-Xu-Zhao} that there are even-like binary $[15m+4,4,8m+2]$, $[15m+7,4,8m+4]$, $[15m+11,4,8m+6]$ and $[15m+14,4,8m+8]$ codes, which contradicts the following Griesmer bounds.
$$\begin{array}{ll}
    d(15m+4,4)\leq 8m+1, &
     d(15m+7,4)\leq8m+3, \\
    d(15m+11,4)\leq8m+5, &
     d(15m+14,4)\leq8m+7.
  \end{array}
$$
Hence
$$\begin{array}{ll}
    d_{so}(31m+6,5)\leq16m, &
    d_{so}(31m+13,5)\leq16m+4, \\
    d_{so}(31m+21,5)\leq16m+8, &
     d_{so}(31m+28,5)\leq16m+12.
  \end{array}
$$

By Table 1 in \cite{SO-40}, there exist binary $[37,5,16]$, $[13,5,4]$, $[21,5,8]$ and $[28,5,12]$ SO codes. By Lemma \ref{lemma-k-3}, there are binary $[31m+6,5,16m]$, $[31m+13,5,16m+4]$, $[31m+21,5,16m+8]$ and $[31m+28,5,16m+12]$ SO codes for $m\geq 1$. Hence
$$\begin{array}{ll}
    d_{so}(31m+6,5)\geq 16m, &
    d_{so}(31m+13,5)\geq16m+4, \\
    d_{so}(31m+21,5)\geq16m+8, &
    d_{so}(31m+28,5)\geq16m+12.
  \end{array}
$$
This implies that
$${\small\begin{array}{l}
    d_{so}(31m+6,5)= 16m=d(31m+6,5)-2, \\
     d_{so}(31m+13,5)=16m+4=d(31m+13,5)-2, \\
    d_{so}(31m+21,5)=16m+8=d(31m+21,5)-2, \\
    d_{so}(31m+28,5)=16m+12=d(31m+28,5)-2.
  \end{array}}
$$
That is to say, $d_{so}(n,5)=d(n,5)-2$ for $n\geq 32$ and $n\equiv 6,13,21,28~({\rm mod}~31)$. This completes the proof.
\end{proof}

\section{Binary optimal $[n,6]$ self-orthogonal codes}

Using a similar approach, we can completely generalize the results of \cite[Theorem 21]{Kim-SO}.

\begin{theorem} {\rm \cite{Kim-SO}}
For lengths $41\leq n\leq 256$, if $n \not\equiv 7,14,22,29,38,45,53,60~({\rm mod\ 63})$ and $n\neq 46,54,61$, then $d_{so}(n,6)=2\left\lfloor \frac{d(n,6)}{2}\right\rfloor$.
\end{theorem}

\begin{theorem}
For lengths $n\geq 41$, if $n \not\equiv 7,14,22,29,38,45,53,60~({\rm mod\ 63})$ and $n\neq 46,54,61$, then $d_{so}(n,6)=2\left\lfloor \frac{d(n,6)}{2}\right\rfloor$.
\end{theorem}

\begin{proof}
We can refer to \cite[Theorem 21]{Kim-SO} for $n < 63$, so we only need to consider $n \geq 63$. The proof is discussed in 16 cases.

(i) It is well-known that the binary Simplex code of dimension $6$ is an $[63,6,32]$ Griesmer SO code. So there is a binary $[63m,6,32m]$ Griesmer SO code for an integer $m\geq 1$.
When $n=63m+6$ and $k=6$, the largest minimum distance $d$ satisfying the Griesmer bound is $d=32m+1$. So
\begin{align*}
  32m\leq d_{so}(63m,6) \leq  d_{so}(63m+6,6)
  \leq  2\left\lfloor \frac{d(63m+6,6)}{2}\right\rfloor
  \leq 2\left\lfloor \frac{d}{2}\right\rfloor=32m.
\end{align*}
It follows from $d_{so}(n,6)\leq d_{so}(n+1,6)$ that $d_{so}(n,6)=2\left\lfloor \frac{d(n,6)}{2}\right\rfloor$ for $n\geq 63$ and $n \equiv 0,1,2,3,4,5,6~({\rm mod\ 63})$.\\

(ii) From \cite{Kim-SO}, there is a binary $[71,6,34]$ SO code, which is optimal with respect to the Griesmer bound. According to Lemma \ref{lemma-k-3}, there are binary $[63m+8,6,32m+2]$ SO codes for $m\geq 1$. When $n=63m+9$ and $k=6$, the largest minimum distance $d$ satisfying the Griesmer bound is $d=32m+3$. So
{\small\begin{align*}
  32m+2\leq d_{so}(63m+8,6)\leq  d_{so}(63m+9,6)
  \leq  2\left\lfloor \frac{d(63m+9,6)}{2}\right\rfloor
  \leq 2\left\lfloor \frac{d}{2}\right\rfloor=32m+2.
\end{align*}}
This implies that $d_{so}(n,6)=2\left\lfloor \frac{d(n,6)}{2}\right\rfloor$ for $n\geq 63$ and $n \equiv 8,9~({\rm mod\ 63})$.\\

(iii) From best-known linear codes (BKLC) database of MAGMA \cite{magma}, there is a binary $[73,6,36]$ SO code, which is optimal with respect to the Griesmer bound. According to Lemma \ref{lemma-k-3}, there are binary $[63m+10,6,32m+4]$ SO codes for $m\geq 1$. When $n=63m+13$ and $k=6$, the largest minimum distance $d$ satisfying the Griesmer bound is $d=32m+5$. So
{\small\begin{align*}
 32m+4\leq d_{so}(63m+10,6)\leq  d_{so}(63m+13,6)
  \leq   2\left\lfloor \frac{d(63m+13,6)}{2}\right\rfloor\
  \leq  2\left\lfloor \frac{d}{2}\right\rfloor=32m+4.
\end{align*}}
It follows from $d_{so}(n,6)\leq d_{so}(n+1,6)$ that $d_{so}(n,6)=2\left\lfloor \frac{d(n,6)}{2}\right\rfloor$ for $n\geq 63$ and $n \equiv 10,11,12,13~({\rm mod\ 63})$.\\

(iv) We start from a binary $[78,6,38]$ SO code (see \cite{Kim-SO}). Combining Lemma \ref{lemma-k-3} and the Griesmer bound, we have
{\small\begin{align*}
  32m+6\leq d_{so}(63m+15,6)\leq  d_{so}(63m+16,6)
  \leq  2\left\lfloor \frac{d(63m+16,6)}{2}\right\rfloor
  \leq 32m+6.
\end{align*}}
This implies that $d_{so}(n,6)=2\left\lfloor \frac{d(n,6)}{2}\right\rfloor$ for $n\geq 63$ and $n \equiv 15,16~({\rm mod\ 63})$.\\

(v) We start from a binary $[80,6,40]$ SO code (see BKLC in \cite{magma}). Combining Lemma \ref{lemma-k-3} and the Griesmer bound, we have
{\small\begin{align*}
  32m+8\leq d_{so}(63m+17,6)\leq   d_{so}(63m+21,6)
  \leq 2\left\lfloor \frac{d(63m+21,6)}{2}\right\rfloor
  \leq 32m+8.
\end{align*}}
It follows from $d_{so}(n,6)\leq d_{so}(n+1,6)$ that $d_{so}(n,6)=2\left\lfloor \frac{d(n,6)}{2}\right\rfloor$ for $n\geq 63$ and $n \equiv 17,18,19,20,21~({\rm mod\ 63})$.\\

(vi) We start from a binary $[86,6,42]$ SO code (see \cite{Kim-SO}). Combining Lemma \ref{lemma-k-3} and the Griesmer bound, we have
{\small\begin{align*}
 32m+10\leq d_{so}(63m+23,6)\leq   d_{so}(63m+24,6)
  \leq 2\left\lfloor \frac{d(63m+24,6)}{2}\right\rfloor
  \leq 32m+10.
\end{align*}}
This implies that $d_{so}(n,6)=2\left\lfloor \frac{d(n,6)}{2}\right\rfloor$ for $n\geq 63$ and $n \equiv 23,24~({\rm mod\ 63})$.

(vii) We start from a binary $[88,6,44]$ SO code (see BKLC in \cite{magma}). Combining Lemma \ref{lemma-k-3} and the Griesmer bound, we have
{\small\begin{align*}
 32m+12\leq d_{so}(63m+25,6)\leq   d_{so}(63m+28,6)
  \leq 2\left\lfloor \frac{d(63m+28,6)}{2}\right\rfloor
  \leq 32m+12.
\end{align*}}
It follows from $d_{so}(n,6)\leq d_{so}(n+1,6)$ that $d_{so}(n,6)=2\left\lfloor \frac{d(n,6)}{2}\right\rfloor$ for $n\geq 63$ and $n \equiv 25,26,27,28~({\rm mod\ 63})$.\\

(viii) We start from a binary $[93,6,46]$ SO code (see \cite{Kim-SO}). Combining Lemma \ref{lemma-k-3} and the Griesmer bound, we have
{\small\begin{align*}
 32m+14\leq d_{so}(63m+30,6)\leq   d_{so}(63m+31,6)
  \leq 2\left\lfloor \frac{d(63m+31,6)}{2}\right\rfloor
  \leq 32m+14.
\end{align*}}
This implies that $d_{so}(n,6)=2\left\lfloor \frac{d(n,6)}{2}\right\rfloor$ for $n\geq 63$ and $n \equiv 30,31~({\rm mod\ 63})$.\\

(ix) We start from a binary $[95,6,48]$ SO code (see \cite{Kim-SO}). Combining Lemma \ref{lemma-k-3} and the Griesmer bound, we have
{\small\begin{align*}
 32m+16\leq d_{so}(63m+32,6)\leq  d_{so}(63m+37,6)
  \leq 2\left\lfloor \frac{d(63m+37,6)}{2}\right\rfloor
  \leq 32m+16.
\end{align*}}
It follows from $d_{so}(n,6)\leq d_{so}(n+1,6)$ that $d_{so}(n,6)=2\left\lfloor \frac{d(n,6)}{2}\right\rfloor$ for $n\geq 63$ and $n \equiv 32,33,34,35,36,37~({\rm mod\ 63})$.\\

(x) We start from a binary $[102,6,50]$ SO code (see \cite{Kim-SO}). Combining Lemma \ref{lemma-k-3} and the Griesmer bound, we have
{\small\begin{align*}
 32m+18\leq d_{so}(63m+39,6)\leq   d_{so}(63m+40,6)
  \leq 2\left\lfloor \frac{d(63m+40,6)}{2}\right\rfloor
  \leq 32m+18.
\end{align*}}
This implies that $d_{so}(n,6)=2\left\lfloor \frac{d(n,6)}{2}\right\rfloor$ for $n\geq 63$ and $n \equiv 39,40~({\rm mod\ 63})$.\\

(xi) We start from a binary $[104,6,52]$ SO code (see BKLC in \cite{magma}). Combining Lemma \ref{lemma-k-3} and the Griesmer bound, we have
{\small\begin{align*}
 32m+20\leq d_{so}(63m+41,6)\leq   d_{so}(63m+44,6)
  \leq 2\left\lfloor \frac{d(63m+44,6)}{2}\right\rfloor
  \leq 32m+20.
\end{align*}}
It follows from $d_{so}(n,6)\leq d_{so}(n+1,6)$ that $d_{so}(n,6)=2\left\lfloor \frac{d(n,6)}{2}\right\rfloor$ for $n\geq 63$ and $n \equiv 41,42,43,44~({\rm mod\ 63})$.\\

(xii) We start from a binary $[109,6,54]$ SO code (see \cite{Kim-SO}). Combining Lemma \ref{lemma-k-3} and the Griesmer bound, we have
{\small\begin{align*}
 32m+22\leq d_{so}(63m+46,6)\leq   d_{so}(63m+47,6)
  \leq 2\left\lfloor \frac{d(63m+47,6)}{2}\right\rfloor
  \leq 32m+22.
\end{align*}}
This implies that $d_{so}(n,6)=2\left\lfloor \frac{d(n,6)}{2}\right\rfloor$ for $n\geq 63$ and $n \equiv 46,47~({\rm mod\ 63})$.\\

(xiii) We start from a binary $[111,6,56]$ SO code (see BKLC in \cite{magma}). Combining Lemma \ref{lemma-k-3} and the Griesmer bound, we have
{\small\begin{align*}
 32m+24\leq d_{so}(63m+48,6)\leq   d_{so}(63m+52,6)
  \leq 2\left\lfloor \frac{d(63m+52,6)}{2}\right\rfloor
  \leq 32m+24.
\end{align*}}
It follows from $d_{so}(n,6)\leq d_{so}(n+1,6)$ that $d_{so}(n,6)=2\left\lfloor \frac{d(n,6)}{2}\right\rfloor$ for $n\geq 63$ and $n \equiv 48,49,50,51,52~({\rm mod\ 63})$.\\

(xiv) We start from a binary $[117,6,58]$ SO code (see \cite{Kim-SO}). Combining Lemma \ref{lemma-k-3} and the Griesmer bound, we have
{\small\begin{align*}
 32m+26\leq d_{so}(63m+54,6)\leq  d_{so}(63m+55,6)
  \leq 2\left\lfloor \frac{d(63m+55,6)}{2}\right\rfloor
  \leq 32m+26.
\end{align*}}
This implies that $d_{so}(n,6)=2\left\lfloor \frac{d(n,6)}{2}\right\rfloor$ for $n\geq 63$ and $n \equiv 54,55~({\rm mod\ 63})$.\\

(xv) We start from a binary $[119,6,60]$ SO code (see BKLC in \cite{magma}). Combining Lemma \ref{lemma-k-3} and the Griesmer bound, we have
{\small\begin{align*}
 32m+28\leq d_{so}(63m+56,6)\leq   d_{so}(63m+59,6)
  \leq 2\left\lfloor \frac{d(63m+59,6)}{2}\right\rfloor
  \leq 32m+28.
\end{align*}}
It follows from $d_{so}(n,6)\leq d_{so}(n+1,6)$ that $d_{so}(n,6)=2\left\lfloor \frac{d(n,6)}{2}\right\rfloor$ for $n\geq 63$ and $n \equiv 56,57,58,59~({\rm mod\ 63})$.\\

(xvi) We start from a binary $[124,6,62]$ SO code (see \cite{Kim-SO}). Combining Lemma \ref{lemma-k-3} and the Griesmer bound, we have
{\small\begin{align*}
 32m+30\leq d_{so}(63m+61,6)\leq  d_{so}(63m+62,6)
  \leq 2\left\lfloor \frac{d(63m+62,6)}{2}\right\rfloor
  \leq 32m+30.
\end{align*}}
This implies that $d_{so}(n,6)=2\left\lfloor \frac{d(n,6)}{2}\right\rfloor$ for $n\geq 63$ and $n \equiv 61,62~({\rm mod\ 63})$.
\end{proof}

\begin{theorem}
If $n\geq 63$ and $n \equiv 7,14,22,29,38,45,53,60~({\rm mod\ 63})$, then $d_{so}(n,6)=d(n,6)-2$, i.e., there are no binary $[n,6,d(n,6)]$ SO codes.
\end{theorem}

\begin{proof}
From the Database \cite{codetables}, there are binary $[70,6,34]$, $[77,6,38]$, $[85,6,42]$, $[92,6,46]$, $[101,6,50]$, $[108,6,54]$, $[116,6,58]$ and $[123,6,62]$ linear codes. Similar to Lemma \ref{lemma-k-3}, there are binary $[63m+7,6,32m+2]$, $[63m+14,6,32m+6]$, $[63m+22,6,32m+10]$, $[63m+29,6,32m+14]$, $[63m+38,6,32m+18]$, $[63m+45,6,32m+22]$, $[63m+53,6,32m+26]$ and $[63m+60,6,32m+30]$ linear codes for $m\geq 1$. Combined with the Griesmer bound, we have
{\small$$
\begin{array}{ll}
  d(63m+7,6)=32m+2, & d(63m+14,6)=32m+6, \\
   d(63m+22,6)=32m+10, &d(63m+29,6)=32m+14, \\
    d(63m+38,6)=32m+18, & d(63m+45,6)=32m+22, \\
  d(63m+53,6)=32m+26, & d(63m+60,6)=32m+30.
\end{array}
$$}

If there are binary $[63m+7,6,32m+2]$, $[63m+14,6,32m+6]$, $[63m+22,6,32m+10]$, $[63m+29,6,32m+14]$, $[63m+38,6,32m+18]$, $[63m+45,6,32m+22]$, $[63m+53,6,32m+26]$ and $[63m+60,6,32m+30]$ SO codes for $m\geq 1$, then it follows from Lemma \ref{lem-Li-Xu-Zhao} that there are even-like binary $[31m+5,5,16m+2]$, $[31m+8,5,16m+4]$, $[31m+12,5,16m+6]$, $[31m+15,5,16m+8]$, $[31m+20,5,16m+10]$, $[31m+23,5,16m+12]$, $[31m+27,5,16m+14]$ and $[31m+30,5,16m+16]$ linear codes for $m\geq 1$, which contradicts the following Griesmer bound.
{\small$$
\begin{array}{ll}
  d(31m+5,5)\leq 16m+1, & d(31m+8,5)\leq 16m+3,\\
   d(31m+12,5)\leq 16m+5, &d(31m+15,5)\leq 16m+7, \\
    d(31m+20,5)\leq 16m+9, & d(31m+23,5)\leq 16m+11, \\
  d(31m+27,5)\leq 16m+13, & d(31m+30,5)\leq 16m+15.
\end{array}
$$}
Hence
{\small
$$
\begin{array}{ll}
  d_{so}(63m+7,6)\leq32m, & d_{so}(63m+14,6)\leq 32m+4, \\
   d_{so}(63m+22,6)\leq 32m+8, &d_{so}(63m+29,6)\leq 32m+12, \\
    d_{so}(63m+38,6)\leq 32m+16, & d_{so}(63m+45,6)\leq 32m+20, \\
  d_{so}(63m+53,6)\leq 32m+24, & d_{so}(63m+60,6)\leq 32m+28.
\end{array}
$$ }

By Table 2 in \cite{Kim-SO}, there are binary $[70,6,32]$, $[77,6,36]$, $[85,6,40]$, $[92,6,44]$, $[101,6,48]$, $[108,6,52]$, $[116,6,56]$ and $[123,6,60]$ SO codes. According to Lemma \ref{lemma-k-3}, there are binary $[63m+7,6,32m]$, $[63m+14,6,32m+4]$, $[63m+22,6,32m+8]$, $[63m+29,6,32m+12]$, $[63m+38,6,32m+16]$, $[63m+45,6,32m+20]$, $[63m+53,6,32m+24]$ and $[63m+60,6,32m+28]$ SO codes for $m\geq 1$.
Hence
{\small
$$
\begin{array}{ll}
  d_{so}(63m+7,6)\geq32m, & d_{so}(63m+14,6)\geq 32m+4, \\
   d_{so}(63m+22,6)\geq 32m+8, & d_{so}(63m+29,6)\geq 32m+12, \\
    d_{so}(63m+38,6)\geq 32m+16, & d_{so}(63m+45,6)\geq 32m+20, \\
  d_{so}(63m+53,6)\geq 32m+24, & d_{so}(63m+60,6)\geq 32m+28.
\end{array}
$$}
It turns out that
{\small
$$
\begin{array}{ll}
  d_{so}(63m+7,6)=32m, & d_{so}(63m+14,6)= 32m+4, \\
   d_{so}(63m+22,6)= 32m+8, & d_{so}(63m+29,6)= 32m+12, \\
    d_{so}(63m+38,6)= 32m+16, & d_{so}(63m+45,6)= 32m+20, \\
  d_{so}(63m+53,6)= 32m+24, & d_{so}(63m+60,6)= 32m+28.
\end{array}
$$}
This implies that $d_{so}(n,6)=d(n,6)-2$ for $n\geq 63$ and $n \equiv 7,14,22,29,38,45,53,60~({\rm mod\ 63})$.
\end{proof}

\section{Conclusion}
In this paper, we solved the two conjectures proposed by Kim and Choi in \cite{Kim-SO} and completely generalized some results of them by considering binary SO codes related to the simplex codes.

\section*{Conflict of Interest}
The authors have no conflicts of interest to declare that are relevant to the content of this
article.

\section*{Data Deposition Information}
Our data can be obtained from the authors upon reasonable request.

\section*{Acknowledgement}
This research is supported by Natural Science Foundation of China (12071001).

\end{document}